\documentclass[12pt,draftcls,onecolumn]{IEEEtran}
\usepackage{amsmath}
\usepackage{amsthm}
\usepackage{amsfonts}
\usepackage{amssymb}
\usepackage[utf8]{inputenc}
\DeclareUnicodeCharacter{00A0}{ }
\usepackage{graphicx}
\usepackage{graphics}
\usepackage{float}
\usepackage{tikz}
\usetikzlibrary{shapes,snakes,patterns}
\usepackage{epstopdf}

\usepackage[justification=centering]{caption}
\usepackage{enumerate} 
\usepackage{cite}
\usepackage{placeins}
\newcommand{\vast}{\bBigg@{4}}
\newcommand{\Vast}{\bBigg@{5}}
\newtheorem{theorem}{Theorem}[]
\newtheorem{lemma}[]{Lemma}
\newtheorem{corollary}{Corollary}[theorem]
\usepackage{suffix}
\usepackage{mathtools}
\usepackage{amsthm}
\usepackage{enumitem}
\usepackage{xfrac}
\usepackage{relsize}
\theoremstyle{definition}

\usepackage[list=true]{subcaption}
\usepackage[normalem]{ulem}
\begin{document}
	\bstctlcite{IEEEexample:BSTcontrol}
	\title{On the Asymptotic Performance Analysis of the k-th Best Link Selection over Non-identical Non-central Chi-square Fading Channels}
	\author{Athira Subhash, Sheetal Kalyani, Yazan H. Al-Badarneh, and Mohamed-Slim Alouini {\let\thefootnote\relax\footnote{{Athira Subhash, and Sheetal Kalyani are with the Department of Electrical Engineering, Indian Institute of Technology, Madras, India. (email:\{ee16d027@smail,skalyani@ee\}.iitm.ac.in).\\Yazan H.Al-Badarneh is with Department of Electrical Engineering, University of Jordan, Amman, Jordan.(email:albadarneh@tamu.edu).\\ M.-S. Alouini is with King Abdullah University of Science and Technology (KAUST), Thuwal, Makkah Province, Saudi Arabia. (e-mail:{slim.alouini}@kaust.edu.sa).}}}
	}
	
	\maketitle
	\begin{abstract}
		\textcolor{black}{This paper derives the asymptotic distribution of the normalized $k$-th maximum order statistics of a sequence of non-central chi-square random variables with non-identical non-centrality parameter. We demonstrate the utility of these results in characterizing the signal to noise ratio in three different applications in wireless communication systems where the statistics of the $k$-th maximum channel power over Rician fading links are of interest. }Furthermore, we derive simple expressions for the asymptotic outage probability,  average throughput, achievable throughput, and the average bit error probability. The proposed results are validated via extensive Monte Carlo simulations.
	\end{abstract}
	
	
	\section{Introduction}
	
	Non-central chi-square random variables (RVs) are widely encountered both in the statistics and the communication literature \cite{mooney1999performance,ding1994efficient,samir1990computation,hemachandra2011novel,bhatnagar2013capacity,qian2020beamforming}. Characterization of the order statistics of non-central chi-square RVs are of interest in several of these applications  \cite{hemachandra2011novel,bhatnagar2013capacity,le2017selection}. The cumulative distribution function (CDF) of a non-central chi-square RV is available in terms of the Marcum-Q function \cite{molisch2012wireless}. Hence the exact order statistics of the non-central chi-square RV are expressed as complicated functions involving the Marcum-Q function. This is particularly the case when the order statistics are evaluated over large sequences of RVs. Extreme value theory (EVT) is a very useful tool to characterise the asymptotic extreme order statistics in such scenarios \cite{falk2010laws,de2007extreme}.
	\par Tools from EVT has been utilized in diverse fields for the asymptotic performance analysis of systems involving extreme order statistics \cite{de2007extreme,al2018asymptotic,song2006asymptotic,subhash2020transmit,xue2015performance,alperformance,falk2010laws,krishnamurthy2019peak,bai2009rate}. All of these works make use of the classical Fischer-Tippet theorem to characterize the normalized extreme order statistics of a sequence of independent and identically distributed (i.i.d.) RVs \cite{falk2010laws}. Although there are results in statistics that extend this classical result to the case of independent and non identically distributed (i.n.i.d.) RVs, this result has not been utilized in the communication theory literature, until recently \cite{subhash2021cooperative}. The classical EVT approach for deriving the statistics of extremes begins by identifying the maximum domain of attraction of the sequences of identical RVs and then identifying the corresponding normalizing constants. The theory for the case of i.n.i.d. RVs proceeds differently and hence provide another interpretation of the extreme order statistics, even for the case of i.i.d. RVs.
	\par In this work, we make use of the results from \cite{barakat2002limit} to characterize the $k$-th maximum order statistics of a sequence of i.n.i.d. non-central chi-square RVs with two degrees of freedom (dof). Even for the case of i.i.d. RVs, the asymptotic order statistics of the maximum RV are not available in closed form. The authors of \cite{xue2015performance} use EVT to characterize the maximum order statistics of the signal to noise ratio (SNR) in a multi-way relaying scenario where the channel experiences i.i.d. Rician fading. However, using the classical EVT approach, it is difficult to derive the parameters of the proposed extreme value distribution in closed form. Hence, they derive the lower and upper bounds of the maximum
	value of the channel gain, which follows the non-central chi-square distribution. {The proposed results can be used to derive exact closed-form expressions for the normalizing constants for the case of i.i.d. RVs as well and hence is observed to be tighter than the results proposed by \cite{xue2015performance}.} 
	\par Several results from the statistics literature have been effectively used to analyze order statistics in the communication literature. The authors of \cite{ikki2010performance} use order statistics to study the statistics of the output SNR in a cooperative relay network with the $k$-th best relay selection scheme. The first and second-order statistics of the signal to interference ratio in a system, where the source signal experiences Rician fading and the interferer signals experience Nakagami-m fading, is studied by the authors of \cite{milic2016second}. The second-order statistics of the channel gain in the more general $\kappa-\mu$ shadowed fading model is used to derive closed-form expressions for the level crossing rate and average fade duration in \cite{cotton2015second}.
	\par In all the above literature, one can notice that the exact statistics of the $k$-th extreme order statistics has a very complicated expression and is challenging to evaluate when the extremes are evaluated over a sequence of a large number of RVs. This is especially the case when each of the RVs in the sequence has a CDF with a complicated expression. Hence, in this work, we derive simple expressions for the CDF of the $k$-th order statistics using results from EVT. Using the results from \cite{barakat2002limit}, we first characterize the asymptotic distribution of the normalized $k$-th maximum SNR in closed form for the case of extremes over a sequence of i.i.d. RVs. {We then characterize the distribution of the normalized maximum SNR when the line of sight (LOS) path corresponding to the different signal links experience statistically non-identical fading.} To the best of our knowledge, this is the first work providing closed-form expressions for the asymptotic distribution of the $k$-th maximum SNR in a Rician fading scenario. {We also demonstrate how the simple form of the asymptotic results proposed can be used to establish stochastic ordering results of the maximum RV. The theory of stochastic ordering quantifies the concept of one RV being bigger than another \cite{shaked2007stochastic}. These ordering results can in-turn, be used to study the variations in system performance with respect to the variations in different system parameters.}
	\par Use of non-terrestrial communication links is considered to be one promising enabler for meeting the requirements of the future generations of communication systems \cite{murali_aaingn,gopi2021cooperative,motlagh2016low}. In the era of IoT, we expect a proliferation of communicating nodes in the network, many of which transmit critical information sporadically. Hence, it is important that these nodes are protected from malicious eavesdroppers and that the information reaches the information processing unit reliably \cite{srinivasan2018secrecy,lei2019safeguarding}. In this regard, the first application that we consider studies a unmanned aerial vehicle (UAV) assisted internet of things (IoT) communication network where the IoT device communicates in the presence of multiple eavesdroppers. In such a scenario, the secrecy is always measured with respect to the strongest eavesdropper link, and hence we are interested in characterizing the maximum eavesdropper SNR. The derived maxima statistics can be used for the performance analysis of such a system. Next, we also study an IoT communication network where multiple UAVs are available for information reception. These UAVs can themselves act as information processing units or can be relay nodes capable of transmitting the data to a common information processing node with high reliability. Given that multiple UAVs are available, the IoT device selects the UAV with the $k$-th best SNR link in each transmission slot. Note that the UAV corresponding to the link with the maximum SNR is always preferred, but if the particular UAV is not available for communication, the IoT device will select the $k$-th best UAV. In such a scenario also, we are interested in characterizing the statistics of the channel power of the $k$-th best link. Furthermore, in both of these scenarios, there is a very high probability for a LOS link between the IoT device and the UAV/eavesdropper, and hence the channel between the nodes are assumed to experience independent Rician fading. Moreover, the different UAVs/eavesdroppers can experience different scattering environments, and therefore these links may not experience identical channel fading conditions. This, in turn, motivates the study of the asymptotic statistics of the received SNR in these scenarios over i.n.i.d. channels. 
	
	\par We also demonstrate the utility of the results derived in characterizing the received SNR of a $k$-th best selection combining (SC) receiver where each signal link experiences independent Rician fading. Diversity combining has long been accepted as an effective method to combat issues like fading, poor coverage, and shadowing \cite{simon2005digital}. The performance of different diversity combining techniques have been extensively studied under different performance metrics in the literature \cite{eng1996comparison,chen2005analysis,gorokhov2003receive,alouini2000mgf,srinivasan2019analysis}. The choice of a particular combining technique depends upon the computational and power constraints of the corresponding system. Out of the different combining techniques, SC is one of the simplest methods with the least computational complexity. Moreover, SC requires only one radio frequency chain at the receiver and hence consumes less power. When compared to a single link between the transmitter and the receiver, the SC system demonstrates better performance with minimal power and computational requirements.
	\par An SC system may select the best link or, more generally, the $k$-th best link in terms of the desired performance metric for information transmission/processing. Statistical analysis of such SC systems has been demonstrated to provide valuable insights that facilitate performance analysis and effective resource allocation \cite{kong2000snr,al2018asymptotic,subhash2019asymptotic}. The proposed results can also be used to analyze these SC systems.  {Hence, in this work, we demonstrate three different applications in the communication literature where we need the statistics of the $k$-th maximum of non-central chi-square RVs for characterizing the system performance. As discussed above, two of these applications are in the field of UAV assisted IoT communication systems, and the third application considers an SC receiver.}
	\color{black}
	\par {The rest of the paper is organized as follows. We present the derivation of the $k$-th maximum order statistics of the non-central chi-square RVs with two dof in section \ref{sec_2a} and Section \ref{sec_2b}. Next, we present stochastic ordering results for the derived maximum distribution in Section \ref{stochastic_ordering}. Using the derived distribution of the $k$-th maximum, we derive expressions for the asymptotic outage probability, average throughput, achievable throughput, and the average bit error probability (BEP) in { Section \ref{analysis_spcl} and Section \ref{analysis_all}}. We then present three different applications of the proposed results in Section \ref{sys_model}. Next, in section \ref{simulations}, we verify the results presented in the paper using simulation experiments, and we conclude the paper in Section \ref{conclusion}.}

	\section{Distribution of the k-th maximum order statistics} \label{sec_kth_max}
	In this section, we derive the distribution of the normalised $k$-th maximum order statistics of a sequence of i.n.i.d. non-central chi-square RVs with two dof. Let $\{Z_m\}_{m=1}^M$ represent a sequence of $M$ non-central chi-square RVs with CDF given by,
	\begin{equation}
	F_{Z_m}(z) = 1-\mathcal{Q}_1 \left(\frac{\nu_m}{\sigma},\frac{\sqrt{z}}{\sigma} \right),
	\label{cdf_zm}
	\end{equation}
	where $\mathcal{Q}_1(,.,)$ is the Marcum-Q function \cite{molisch2012wireless}. {Here, $\frac{\nu_m}{\sigma}$ is the non-centrality parameter of the non-central chi square RV\footnote{{Note that $F_{Z_m}(z)$ also represents the CDF of the square of a Rician RV with shape parameters $\nu_m$ and $\sigma$. According to the Rician channel fading model this corresponds to the scenario where $\nu_m^2$ is the power in the direct path and $2\sigma^2$ is the total power in the scattered paths.}}}. The corresponding order statistics are then given by $Z_{ (1: M)} \leq Z_{(2: M)} \leq \cdots \leq Z_{(M:M)}$, where the $k$-th order statistic is given by $Z_{(M-k+1:M)}$.
	The exact CDF of $Z_{(M-k+1:M)}$, is given by \cite[(5.2.1)]{david2003order}
	\begin{equation}
	F_{Z_{(M-k+1:M)}}(z)=\sum_{m=k}^{M} \sum_{S_{m}} \prod_{r=1}^{m} F_{Z_{j_r}}(z) \prod_{r=m+1}^{M}\left[1-F_{Z_{j_r}}(z)\right], \ k=1,2,\cdots,M,
	\label{cdf_kth_max_exact}
	\end{equation}
	where the summation $S_{m}$ is over all the permutations $\left(j_{1}, \ldots, j_{M}\right)$ of $1, \ldots, M$ for which $j_{1}<\cdots<j_{m}$ and $j_{m+1}<\cdots<j_{M}$. Note that the exact CDF of $Z_{(M-k+1:M)}$ in (\ref{cdf_kth_max_exact}) will have a very complicated expression when each of the RV $Z_m$ follows the distribution given in (\ref{cdf_zm}). Furthermore, it will not be easy to use this expression for deriving the statistics of other functions of the $k$-th order statistic or to study the effect of the different parameters on the order statistics. Hence in this section, we use tools from EVT to derive the distribution of $Z_{(M-k+1:M)}$ for large $M$. We first derive the distribution of the normalized $k$-th maximum RV $\tilde{Z}_{(M-k+1:M)} = \frac{Z_{(M-k+1:M)}-b_M}{a_M}$ for normalizing constants $a_M$ and $b_M$. Using these results, we characterize the distribution of $Z_{(M-k+1:M)}$. It should be noted that D.G. Mejzler had studied the order statistics of a sequence of i.n.i.d. RVs in \cite{mejzler1969some}. Then, the authors of \cite{barakat2002limit} utilized this result as well as few other related literature \cite{weissman1975extremal,juncosa1949asymptotic} to study the limiting behaviour of the extreme order statistics of $M$ two-dimensional i.n.i.d. random vectors. The key theorem they utilize is reproduced here for the case of uni-variate RVs. Necessary and sufficient conditions under which the distribution of $\tilde{Z}_{(M-k+1:M)}$ converges weakly to a non degenerate limit, as well as the form of this limit, are presented in {Lemma \ref{thm_order_stat}} .
	
	\subsection{On the Asymptotic Distribution of the normalized $k$-th Maximum}\label{sec_2a}
	
	\begin{lemma} \label{thm_order_stat} \color{black}
		Assume that for suitable normalizing constants $a_{M}>0,$ $b_{M}$
		\begin{equation}
		\delta_{M}=\max _{1 \leq m \leq M} 1-F_{Z_m} \left(a_Mz+b_M\right) \rightarrow 0 \text { as } M \rightarrow \infty.
		\label{ua1}
		\end{equation}
		Then $\tilde{\phi}_{k:M}(z) = \mathbb{P}\left(\frac{Z_{(M-k+1:M)}-b_M}{a_M} \leq z \right)$ converges weakly to a non degenerate distribution function $\tilde{\phi}_{k}\left(z\right)$ if and only if, for all $z$ for which $\tilde{\phi}_{k}\left(z\right)>0$, the limit
		\begin{equation}
		\tilde{u}(z) = \lim _{M \rightarrow \infty} \sum_{m=1}^{M} 1-F_{Z_m}(a_Mz+b_M) \ \text{is finite,}
		\label{ua2}
		\end{equation}
		and the function
		\begin{equation}
		\tilde{\phi}_{k}\left(z\right) = \sum\limits_{m=0}^{k-1} \frac{\tilde{u}^m(z)}{m!} \exp (-\tilde{u}(z)), \ \text{is a non degenerate distribution.}
		\label{ua3}
		\end{equation}
		The actual limit of $\tilde{Z}_{(M-k+1:M)}=\frac{Z_{(M-k+1: M)}-b_{M}}{a_{M}} $ is the one given in (\ref{ua3}). 
	\end{lemma}
	
	\begin{proof}
		Please refer \cite{barakat2002limit} for the detailed proof.
	\end{proof}
	\color{black}
	Furthermore, they make the following observations. The convergence of $\tilde{\phi}_{k:M}\left(z\right)$ for at least one fixed value of $k$ implies its convergence for all fixed values of $k$. Note that we can recover the results for the maximum RV using these results for $k=1$. {Using Lemma \ref{thm_order_stat}, we identify normalizing constants $a_M$ and $b_M$ and hence characterize the distribution of $\tilde{Z}_{(M-k+1:M)}$ for the special case where the parameter $\nu_m$ are non-identical across the users and the parameter $\sigma$ is identical across the users. Note that in a Rician fading environment, this special case corresponds to the scenario where different users receive unequal power along the direct LOS path and equal sum power across all the scattered paths.} The corresponding results are presented in the following theorem.
	
	\subsection{$k$-th Maximum of Non-Central Chi Square RVs }\label{sec_2b}
	\begin{theorem} \label{thm_kth_max}
		The asymptotic CDF of the normalized $k$-th maximum $\left(\tilde{Z}_{(M-k+1:M)}\right)$ of a sequence of non-central chi-square RVs with CDF as given in (\ref{cdf_zm}) is given by,
		\begin{equation}
		F_{\tilde{Z}_{(M-k+1:M)}}(z) = \frac{\Gamma(k,p \exp(-z))}{\Gamma(k)},
		\label{k_th_max_cdf}
		\end{equation}
		for the choice of normalizing constants $a_M=2 \sigma^2$ and 
		\begin{equation}
		b_M=a_M\left(\log(\tilde{M}) - \frac{\log(\log(\tilde{M}))}{4}  + \frac{\tilde{\nu} \sqrt{2}}{\sigma} (\log(\tilde{M}))^{\frac{1}{2}}  - \frac{1}{2} \log \left( \frac{2 \sqrt{2}\pi \tilde{\nu}}{\sigma}\exp\left(\frac{-\tilde{\nu}^2}{\sigma^2} \right) \right)\right).
		\label{b_M_gen}
		\end{equation}
		{Here, $p=\sum\limits_{i=1}^K \frac{M_i}{\tilde{M}} \sqrt{\frac{\tilde{\nu}}{\nu_i}} \exp \left(\frac{-\nu_i^2-\tilde{\nu}^2}{2 \sigma^2} \right)  \exp\left(\frac{\sqrt{2}}{\sigma} \left((\nu_i-\tilde{\nu})\sqrt{\log(\tilde{M})}+\frac{\nu_i\tilde{\nu}}{\sqrt{2} \sigma} \right)\right)$. Furthermore, $\tilde{\nu}$ is the largest element in $ R_{\nu}$ such that 
			\begin{equation}
			\lim\limits_{M \to \infty} \frac{M_i}{\tilde{M}} \exp \left(\frac{\sqrt{2\log(\tilde{M}{})}(\nu_i-\tilde{\nu})}{\sigma} \right) < \infty \quad  \text{for} \quad \nu_i\neq \tilde{\nu},
			\label{m_condn1}
			\end{equation}
			where $\tilde{M} = \sum\limits_{m=1}^M \mathbb{I}_{\nu_m=\tilde{\nu}}$, $M_i := \sum\limits_{m=1}^M \mathbb{I}_{\nu_m=\nu_i}, \ 1\leq i \leq K$, i.e $M_i$ represents the number of times $\nu_i$ occurs among the $M$ values $\{\nu_m\}_{m=1}^M$, and $R_{\nu}$ is a set of finite cardinality which contains all the possible values of the parameter $\nu_m$. }
	\end{theorem}
	\begin{proof}
		Please refer Appendix \ref{proof_main} for the proof.
	\end{proof}
	Here, $\Gamma(a,x)$ is the upper incomplete gamma function \cite{upper_gamma} and $\Gamma(x)$ is the gamma function \cite{gamma} where $\Gamma(a)=\Gamma(a,0)$. Next, we present two special cases of the above result in Corollary \ref{corollary_iid} and Corollary \ref{corollary_iid2}. 
	
	\begin{corollary} \label{corollary_iid}
		The asymptotic CDF of the normalized $k$-th maximum of a sequence of $M$ i.i.d. non-central chi-square RVs with two dof is given by 
		\begin{equation}
		F_{\tilde{Z}_{(M-k+1:M)}}(z) = \frac{\Gamma(k,\exp(-z))}{\Gamma(k)},
		\label{k_th_max_cdf_iid}
		\end{equation}
		for the choice of normalizing constants $a_M=2 \sigma^2$ and 
		\begin{equation}
		b_M=a_M\left(\log({M}) - \frac{\log(\log({M}))}{4}  + \frac{{\nu} \sqrt{2}}{\sigma} \sqrt{\log({M})} - \frac{1}{2} \log \left( \frac{2 \sqrt{2}\pi {\nu}}{\sigma}\exp\left(\frac{-{\nu}^2}{\sigma^2} \right) \right)\right).
		\label{b_M_gen_iid}
		\end{equation}
	\end{corollary}
	\begin{proof}
		This result can be derived by substituting $\nu_m=\nu$ for all $m$ in Theorem \ref{thm_kth_max}.
	\end{proof}
	
	\begin{corollary}\label{corollary_iid2}
		The asymptotic CDF of the normalized maximum of a sequence of non-central chi-square RVs with two dof and non-identical non-centrality parameter is given by
		\begin{equation}
		F_{\tilde{Z}_{(M:M)}}(z) = \exp(-p\exp(-z)),
		\end{equation}
		for the same choice of normalizing constants proposed in Theorem \ref{thm_kth_max}.
	\end{corollary}
	\begin{proof}
		This result can be derived by substituting $k=1$ in Theorem \ref{thm_kth_max}.
	\end{proof}
	{The authors of \cite{xue2015performance} also use EVT to study the maximum order statistics of i.i.d. non-central chi-square RVs with two degrees of freedom and prove that the limiting distribution is the Gumbel distribution. However, they derive bounds on the maximum RV without characterizing the exact values of the normalizing constants. This bound was observed to be loose compared to the results proposed here in simulations.} Given that we have characterized the CDF of the normalized $k$-th maximum RV, the distribution of the $k$-th maximum RV for large values of $M$ can be evaluated as
	\begin{equation}
	F_{Z_{(M-k+1:M)}}(z) =  F_{\tilde{Z}_{(M-k+1:M)}}\left(\frac{z-b_M}{a_M} \right).
	\label{unnorm_cdf_1}
	\end{equation}
	One can observe that the above distribution can also be evaluated as,
	\begin{equation}
	F_{Z_{(M-k+1:M)}}(z) = \sum\limits_{m=0}^{k-1} \frac{u^m(z)}{m!} \exp(-u(z)),
	\label{unnorm_cdf_2}
	\end{equation} where $u(z)=\sum\limits_{m=1}^M \mathcal{Q}_1\left( \frac{\nu_m}{\sigma_m}, \frac{\sqrt{z}}{\sigma_m} \right)$\footnote{Note that we can handle the case of non-identical parameters $\{\sigma_m\}_{m=1}^M$ here.}. The simulations demonstrate that (\ref{unnorm_cdf_2}) can be used to evaluate the CDF of the maximum for moderate values of $M \ (>10 )$. However for larger values of $M$, it is more efficient to evaluate the CDF expression proposed in (\ref{unnorm_cdf_1}). Furthermore, the distribution of the $k$-th maximum given by (\ref{unnorm_cdf_1}) is far easier to evaluate when compared to the exact distribution of the $k$-th maximum evaluated using (\ref{cdf_kth_max_exact}) for $F_{Z_m}(z)$ given by (\ref{cdf_zm}).

	\subsection{Stochastic Ordering Results} \label{stochastic_ordering}
	
	Stochastic ordering is a widely used tool for establishing the ordering of one RV with respect to another \cite{shaked2007stochastic}. Such an ordering proves to be useful in scenarios where we want to study the variations of a RV with respect to the variations in its parameters. In this sub-section, similar to the analysis in \cite[Section III.B]{subhash2021cooperative}, we make use of results from stochastic ordering to see how we can establish the ordering of the $k$-th maximum RV with respect to the variations in different parameters. Here, the $k$-th maximum RV $X$ with parameters $a_M^{(1)}$ and $b_M^{(1)}$ is said to be stochastically smaller than the RV $Y$ with parameters $a_M^{(2)}$ and $b_M^{(2)}$ if 
	\begin{align}
	\mathbb{P}(Y>z) & \geq \mathbb{P}(X>z) \ \forall \ z \ \in \mathbb{R}. \\ 
	\text{i.e} \ \Gamma\left(k,p \times \exp\left(\frac{-\left(\frac{z_{th}}{\gamma_s} \right)+b_M^{(2)}}{a_M^{(2)}} \right) \right) & \leq \Gamma\left(k,p \times \exp\left(\frac{-\left(\frac{z_{th}}{\gamma_s} \right)+b_M^{(1)}}{a_M^{(1)}} \right) \right).
	\label{order_condn1}
	\end{align}
	The inequality in (\ref{order_condn1}) can be true only if the argument of the incomplete Gamma function on the L.H.S is larger than the argument of the incomplete gamma function in the R.H.S of the equation. Upon further simplification, the condition in (\ref{order_condn1}) can be equivalently represented as 
	\begin{equation}
	\frac{\tilde{z}-b_M^{(2)}}{a_M^{(2)}} \leq \frac{\tilde{z}-b_M^{(1)}}{a_M^{(1)}},
	\label{order_condn2}
	\end{equation}where $\tilde{z}=\frac{z_{th}}{\gamma_s}$. Given that we have closed form expressions for the normalizing constants, any change in the parameters can now be mapped to the corresponding change in the normalizing constants $\{a_M^{(i)},b_M^{(i)};i\in \{1,2\}\}$ and hence the corresponding ordering can be established. Note that the condition in (\ref{order_condn2}) is simple to evaluate and it is not possible to arrive at such a simple condition if we were using the exact order statistics of the $k$-th maximum RV. 
\subsection{\textcolor{black}{Other metrics of interest}} \label{analysis_spcl}
In this sub-section, we derive the expressions for the asymptotic probability of outage, average throughput, effective throughput, and the average bit error rate when the received SNR is the $k$-th maximum order statistic evaluated over a sequence of i.n.i.d. RVs with CDF (\ref{cdf_zm}). Throughout this section, we assume that the SNR is given by $\gamma_s Z_{(M-k+1:M)}$ and the sequence of RVs over which the maximum is evaluated belongs to either of the special cases discussed in Section \ref{kth_max_iid_derivation}-\ref{kth_max_inid_derivation}.
	
	\subsubsection{Asymptotic outage probability} \label{op_spcl}
	The probability of outage at the receiver for a threshold of $z_{th}$ is given by
	\begin{equation}
	\mathbb{P}(\gamma_s Z_{(M-k+1:M)} \leq z_{th}) = F_{Z_{(M-k+1:M)}}\left(\frac{z_{th}}{\gamma_s}\right)= \frac{\Gamma\left(k,p \times \exp\left(\frac{-\left(\frac{z_{th}}{\gamma_s} \right)+b_M}{a_M} \right) \right)}{\Gamma(k)},
	\end{equation}
	where $a_M$ and $b_M$ can be evaluated using (\ref{norm_const_iid}) and (\ref{norm_const_inid}) for the case of i.i.d. and i.n.i.d. RVs respectively.
	
	\subsubsection{Asymptotic average throughput} \label{avg_spcl}
	
	The average throughput at the receiver is given by,
	\begin{equation}
	\bar{R}_{(M-k+1:M)}=\mathbb{E}[R_{(M-k+1:M)}] = \mathbb{E}\left[ \log_2(1+\gamma) \right] = \mathbb{E}\left[ \log_2(1+\gamma_s Z_{(M-k+1:M)}) \right].
	\label{erg_capacity}
	\end{equation}
	We know that the RV $\tilde{Z}_{(M-k+1:M)}$ converges in distribution to a RV $\tilde{Z}_k$. Now, let $h(\tilde{Z}_{(M-k+1:M)}):= \log_2(1+\gamma_s (a_M \tilde{Z}_{(M-k+1:M)} + b_M) )$. To evaluate $\mathbb{E}[h(\tilde{Z}_{(M-k+1:M)})]$, we utilize the continuous mapping theorem \cite{billingsley2013convergence}. Thus, we have $h(\tilde{Z}_{(M-k+1:M)}) \xrightarrow[\text{}]{\text{D}} h(\tilde{Z}_{k})$.
	Since $h(\tilde{Z}_{(M-k+1:M)})$ is a positive RV, the expectation of this RV can be evaluated as 
	\begin{equation}
	\mathbb{E}[h(\tilde{Z}_{(M-k+1:M)})] = \int\limits_0^\infty \mathbb{P}(h(\tilde{Z}_{(M-k+1:M)})>x) \ \text{dx}.
	\end{equation}
	Thus, we have
	\begin{equation}
	\lim\limits_{M \to \infty} \mathbb{E}[h(\tilde{Z}_{(M-k+1:M)})] = \lim\limits_{M \to \infty} \int\limits_0^\infty \mathbb{P}(h(\tilde{Z}_{(M-k+1:M)})>x) \ \text{dx}.
	\end{equation}
	Next, we make use of monotone convergence theorem \cite{billingsley2008probability} to simplify the above expression and we have,
	\begin{equation}
	\lim\limits_{M \to \infty} \int\limits_0^\infty \mathbb{P}(h(\tilde{Z}_{(M-k+1:M)})>x) \ \text{d}x = \int\limits_0^\infty \lim\limits_{M \to \infty} \mathbb{P}(h(\tilde{Z}_{(M-k+1:M)})>x) \ \text{d}x = \int\limits_0^\infty \mathbb{P}(h(\tilde{Z}_{k})>x) \ \text{d}x.
	\end{equation}
	Thus, we have $\lim\limits_{M \to \infty} \mathbb{E}[h(\tilde{Z}_{(M-k+1:M)})]=\mathbb{E}[h(\tilde{Z}_k)]$. Given that we have proved the convergence of the moments of $h(\tilde{Z}_{(M-k+1:M)})$ to the moments of $h(\tilde{Z}_{(k)})$, in the following theorem, we propose a series expression to evaluate the average throughput. 
	\begin{theorem} \label{rate_series_thm}
		The average throughput can be evaluated as
		\begin{equation}
		R_{(M-k+1:M)} = \frac{p^k}{\Gamma(k)} \sum\limits_{n=0}^\infty \frac{(-1)^n p^n}{n! (k+n)} \exp \left(\frac{(k+n)(1+\gamma_sb_M)}{\gamma_s a_M} \right) \ E_1\left(\frac{k+n}{a_M \gamma_s} \right).
		\label{sum_expr}
		\end{equation}
	\end{theorem}
	\begin{proof}
		Please refer Appendix \ref{proof_rate} for the proof.
	\end{proof}
	Here, $E_1(.)$ is the exponential integral function \cite{exp_integral_1}. Now, the average throughput can be evaluated using the above infinite summation truncated to a finite number of terms\footnote{Note that for large values of the normalizing constant $b_M$, evaluation of this series expression is observed to have inaccuracies due to the numerical issues while handling large numbers. In such scenarios, we can evaluate the average throughput by directly evaluating (\ref{intg_expr}). Using the numerical integration routines available in software like Matlab and Mathematica, (\ref{intg_expr}) can be evaluated within fractions of a second. Nevertheless, the series expression is useful for deriving inferences about the system performance with respect to the variations in different parameters.}. Note that the expression for average throughput in (\ref{sum_expr}) is more amenable for analysis when compared to the expression in (\ref{intg_expr}). For instance, closed-form expressions for the variables like $p,b_M$ and $a_M$ can be used to infer their variations with respect to the changes in the channel fading parameters. This, in turn, can be used for predicting the corresponding variations in the average throughput by analyzing the corresponding changes in (\ref{sum_expr}). Furthermore, we can use the following lemma from \cite{shaked2007stochastic}, to extend the ordering results in section \ref{stochastic_ordering} to the ordering of the average throughput $\bar{R}_{(M-k+1:M)}$.
	
	\begin{lemma} \label{sto_order}
		RV $X$ is stochastically less than or equal to RV $Y$ if and only if the following holds for all increasing functions $\phi(.)$ for which the expectations exist: $\mathbb{E}[\phi(X)] \leq \mathbb{E}[\phi(Y)]$.
	\end{lemma}
	This aids us in deriving inferences regarding the changes in the average throughput with respect to variations in the system parameters. Note that such observations are otherwise difficult to be derived directly from the exact expressions for the average throughput.
	
	\color{black}
	\subsubsection{Effective Throughput} \label{eff_spcl}
	Under the k-th best link selection scheme, the effective throughput that can be supported under a statistical QoS constraint described by the delay QoS exponent $\zeta$ is given by $\bar{\alpha}_{(M-k+1:M)}$ \cite{al2018asymptotic,wu2003effective},
	\begin{equation}
	\bar{\alpha}_{(M-k+1:M)}=-\frac{1}{\theta} \log _{2}\left(\mathbb{E}\left[e^{-\theta \ln (2) R_{(M-k+1:M)}}\right]\right),
	\label{effec_throughput}
	\end{equation} where $\theta=\zeta T$ and $T$ is the transmission block length. Here the expectation in (\ref{effec_throughput}) is evaluated over the distribution of $Z_{(M-k+1:M)}$ and is given by,
	\begin{align}
	\mathbb{E}\left[e^{-\theta \ln (2) R_{(M-k+1:M)}}\right] & = \mathbb{E} \left[\exp \left(-\theta \ln 2 \log_2(1+\gamma_s Z_{(M-k+1:M)}) \right) \right] \\ &  = \mathbb{E} \left[\exp \left(-\theta \log(1+\gamma_s Z_{(M-k+1:M)})\right) \right]\\
	& = \mathbb{E} \left[(1+\gamma_s Z_{(M-k+1:M)})^{-\theta} \right] \approx \mathbb{E}\left[(\gamma_s Z_{(M-k+1:M)})^{-\theta} \right],
	\label{effec_throughput_integral}
	\end{align} where the last approximation holds for large values of $Z_{(M-k+1:M)})$ (This will in fact be a lower bound on the achievable throughput.). The expectation in (\ref{effec_throughput_integral}) can be evaluated using numerical integration methods. 
	
	\subsubsection{Average BEP} \label{avg_ber_spcl} 
	
	In this section, we consider the class of modulation schemes whose average conditional BEP can be expressed 
	\begin{equation}
	\bar{P}_{e}=C \  \mathbb{E}[\exp({-\rho \gamma})],
	\label{avg_bep_gen}
	\end{equation}
	where $C$, $\rho$ are non-negative and constant for a particular modulation scheme and $\gamma$ is the SNR at the receiver. Thus, we have
	\begin{equation}
	\bar{P}_{(M-k+1:M)}=C \  \mathbb{E}[\exp({-\rho \gamma_s Z_{(M-k+1:M)}})] = C \mathcal{M}_{Z_{(M-k+1:M)}}\left(-\rho\gamma_s \right),
	\label{avg_bep}
	\end{equation}   
	which $\mathcal{M}_{Y}(t)$ is the moment generating function (MGF) of the RV $Y$ evaluated at $t$. Similar to the analysis in \cite{al2018asymptotic}, using Theorem 2 from \cite{chareka2008converse} we can ensure the convergence of the MGF of the RV $\tilde{Z}_{(M-k+1:M)}$ to the MGF of the RV with CDF given in (\ref{k_th_max_exprsn}). For all the special cases mentioned in Section \ref{sec_kth_max}, the MGF of the RV $\tilde{Z}_{(M-k+1:M)} = \frac{{Z}_{(M-k+1:M)}-b_M}{a_M}$ is given by $\mathcal{M}_{\tilde{Z}_{(M-k+1:M)}}(t)=\frac{\Gamma(k-t)}{\Gamma(k)}$. Now, using the relation $\mathcal{M}_{\alpha X+\beta}(t)=\exp(\beta t) \mathcal{M}(\alpha t)$, we have 
	\begin{equation}
	\bar{P}_{(M-k+1:M)} = C \exp(-b_M \rho \gamma_s) \mathcal{M}_{\tilde{Z}_{(M-k+1:M)}}(-a_M \rho \gamma_s) = C  \exp(-b_M \rho \gamma_s)  \frac{\Gamma(k+a_M\rho \gamma_s)}{\Gamma(k)}.
	\end{equation}
	\par If not for the simple expression proposed for the CDF of the $k$-th maximum RV, the evaluation of the performance metrics discussed in Sections \ref{op_spcl}-\ref{avg_ber_spcl} would not have been accessible using the exact order statistics. Using the exact order statistics, we will end up with complicated expressions that does not admit an easy analysis and costly to evaluate.

	\subsection{Analysis for the Other Cases} \label{analysis_all}
	
	Now, for all the cases that does not belong to either of the special cases discussed in the Section \ref{sec_2b}, we can use the CDF of $Z_{(M-k+1:M)}$ given in (\ref{unnorm_cdf_2}) for evaluating the different performance metrics. The corresponding probability density function is obtained by evaluating the derivative of (\ref{unnorm_cdf_2}) and is given by, 
	\begin{equation}
	f_{Z_{(M+k-1:M)}}(z) = u'(z) \left \lbrace \frac{\Gamma(k-1,u(z))}{\Gamma(k-1)}- \frac{\Gamma(k,u(z))}{\Gamma(k)}\right \rbrace,
	\end{equation}where $u'(z)=\frac{du(z)}{dz}$. Even though we cannot have simple closed form expressions for the different performance metrics as derived in the previous sub-section, the above distribution function can be easily used for evaluating the corresponding numerical integral expressions. Moreover, this expression is more useful than the exact distribution of the $k$-th maximum SNR for deriving inferences regarding the system performance under variations in the system parameters.

	\section{Applications} \label{sys_model}
	
	Note that the proposed asymptotic distribution can be utilized in any application where one needs the order statistics of a sequence of non-central chi-square RVs with CDF as given in (\ref{cdf_zm}). This section presents three different applications of the proposed results in communication theory. Apart from the applications presented here, the proposed results can also be utilized to generalize the results presented in \cite{xue2015performance,song2006asymptotic,al2018asymptotic} where the order statistics of the $k$-th maximum channel gain is used for performance analysis. 
	
	\subsubsection{Strongest eavesdropper characterization in IoT systems}
	
	In this sub-section, we consider a UAV assisted IoT system similar to the system model considered by the authors of \cite{lei2019safeguarding}. Here we have a single antenna IoT device $S$ communicating confidential information to a single antenna UAV in the presence of $M$ eavesdroppers represented as $\{E_m\}$. Here we assume that all the terrestrial links experience Rician fading such that they receive non-identical powers along the direct LoS paths and identical power along the scattered components. Hence, the SNR $\gamma_{E,m}$ of the channel between the device $S$ and the eavesdropper $E_m$ is given as follows:
	\begin{equation}
	\gamma_{E_m} = \gamma_s|g_m|^2,
	\end{equation}
	where $\gamma_s$ is the ratio of transmit power and the noise power at the receiver and $g_m$ is the complex channel coefficient of the link between $S$ and the $m$-th eavesdropper. We assume that $\gamma_{E,m}$ follows the non-central chi square distribution with CDF as given in (\ref{cdf_zm})\footnote{{The authors of \cite{lei2019safeguarding} assume that the links between $S$ and each $E_m$ experiences i.i.d. small scale fading but different path losses owing to the difference in distances of the links. However, our model assumes that each $E_m$ experiences similar path losses but non-identical small scale fading environments. This can happen in a scenario where all the eavesdroppers are present at almost the same distances around the source node, but owing to the difference in the scattering environment in different directions experience i.n.i.d. small scale fading channel gains.}}. Now, in the presence of multiple eavesdroppers, the secrecy is always defined with respect to the strongest eavesdropper \cite{alsadi2018improving,alotaibi2014secrecy,kampeas2016secrecy,lei2019safeguarding}. Hence, we are interested in characterizing the statistics of the maximum eavesdropper SNR defined as 
	\begin{equation}
	\gamma_{E,M:M} = \max \{\gamma_{E,1},\cdots, \gamma_{E,M}\}. 
	\label{max_e_snr}
	\end{equation}
	Note that the distribution of $\gamma_{E,M:M}$ can be characterized using the results in Corollary \ref{corollary_iid2}. This would, in turn, facilitate the easy characterization of the secrecy performance of the system. 
	
	\subsubsection{UAV selection} 
	In this section, we consider another UAV assisted IoT communication network. Here, there are $M$ UAVs represented as $\{U_m\}_{m=1}^M$, available in the vicinity of an IoT device $I$ which transmits certain status update sporadically. All the nodes are assumed to be equipped with a single antenna for communication. Let $g_m$ represent the Rician fading channel coefficient between the IoT device and the $m$-th UAV. 
	For each transmission, the IoT device connects to that UAV for which the channel between the nodes $I$ and $U_m$ has the $k$-th largest SNR, i.e the IoT device selects the UAV $\hat{m}$ such that $\gamma_{\hat{m}}=\gamma_{M-k+1:M}$. 
	Here, $\gamma_m=\gamma_s |g_m|^2$ is the received SNR, and $\gamma_s$ is the ratio of transmit power and the noise power at the $m$-th UAV. 
	The IoT device is assumed to have a direct LoS channel with each of the UAVs, and hence all the ground to air links are modeled as Rician faded. Thus, the SNR will follow the non-central chi-square distribution (with CDF as given in (\ref{cdf_zm})) and we are interested in the $k$-th maximum statistics of the SNR evaluated over the $M$ UAV links to characterize the performance. Similar to models used by the authors of \cite{lei2019safeguarding,srinivasan2019analysis}, we assume that the Rician factor $K_m$ of each of the ground to air links is a function of the elevation angle. When the channel between the IoT source and the UAV experiences Rician fading with parameters $K_m$ and $\Omega_m$, the SNR of this link will follow non-central chi-square fading. The Rician fading channel parameters are related to the parameters of the non-central chi-square distribution as follows:
	
	\begin{align}
	\nu_m^{2}&=\frac{K_m\Omega_m}{1+K_m}  \quad  \text{and} \\
	\sigma_m^{2}&=\frac{\Omega_m}{2(1+K_m)}.
	\end{align}
	
	Note that for non-identical values of $K_m$, both the parameters $\nu_m$ and $\sigma_m$ will be non-identical for all the users unlike the case considered in (\ref{cdf_zm}) where only parameter $\nu_m$ is non-identical across the users. To proceed further with the analysis available in Section \ref{sec_kth_max}, we assume that the ratio $\frac{1+K_m}{\Omega_m}:=\theta$ is a constant for all the users\footnote{{Recall that this corresponds to the scenario where all the users receive non-identical power along the dominant LoS path and the total power received along all the scattered paths are identical. This is especially the case in scenarios where the power along the scattered components are negligible.}}. With this assumption, we can arrive at the limiting distribution of the $k$-th maximum SNR following steps very similar to the derivation in Appendix \ref{proof_main}. Here also we can proceed by choosing $a_M=\frac{1}{\theta}$ and $b_M$ as given in (\ref{b_M_gen}). Thus, we can now characterize the performance of the system in terms of the simple expressions proposed for the asymptotic statistics of the $k$-th maximum RV. 
	
	\color{black}
	\subsubsection{Selection combining receiver}
	In this section, we demonstrate the utility of the results derived for the performance analysis of a $k$-th best selection combining receiver in a Rician fading environment. Here we consider a single antenna transmitter node and a single antenna receiver node, each capable of half-duplex communication. There exist $M$ independent communication links between the transmitter and the receiver, and the receiver makes use of the signal along only one of these links for information processing\footnote{These $M$ links may correspond to different relay links or different antennas or any other architecture making use of selection diversity.}. Here, we study the case where the receiver selects the $k$-th best (out of $M$) link in terms of the SNR at the receiver for information processing. Let $g_m$ represent the channel gain across the $m$-th transmitter to receiver link. Furthermore, we assume that each of the $M$ links experiences independent Rician fading such that the CDF of the channel power is given by (\ref{cdf_zm}). Let $P$ be the transmit power and $\delta$ be the noise power at the receiver. Let, $\gamma_s=\frac{P}{\delta}$, then the SNR at the receiver can be expressed as 
	\begin{equation}
	\gamma = \gamma_s |g_{(M-k+1:M)}|^2.
	\end{equation}
	Again, we make use of the results derived in Section \ref{sec_kth_max}  to study the asymptotic probability of outage, average throughput, effective throughput, and the average BEP at the receiver node.


	\section{Simulation Results} \label{simulations}
	
	\begin{figure}[ht]
		\begin{minipage}[b]{0.45\linewidth}
			\centering
			\includegraphics[scale=0.45]{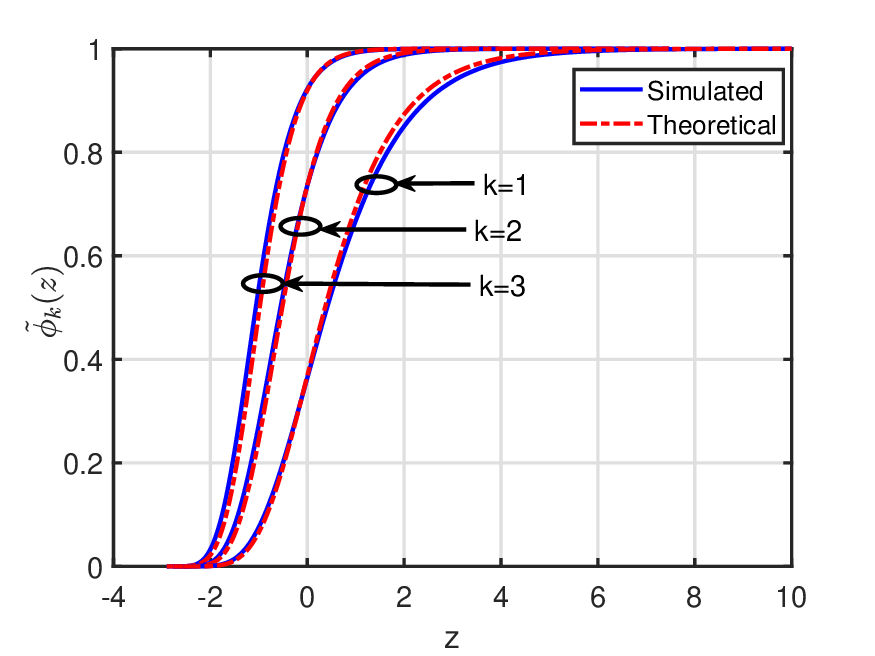}
			\caption{Asymptotic CDF of $\tilde{Z}_{(M-k+1:M)}$ for $M=20$.}
			\label{k_th_max_iid_norm}
		\end{minipage}
		\hspace{0.5cm}
		\begin{minipage}[b]{0.45\linewidth}
			\centering
			\includegraphics[scale=0.45]{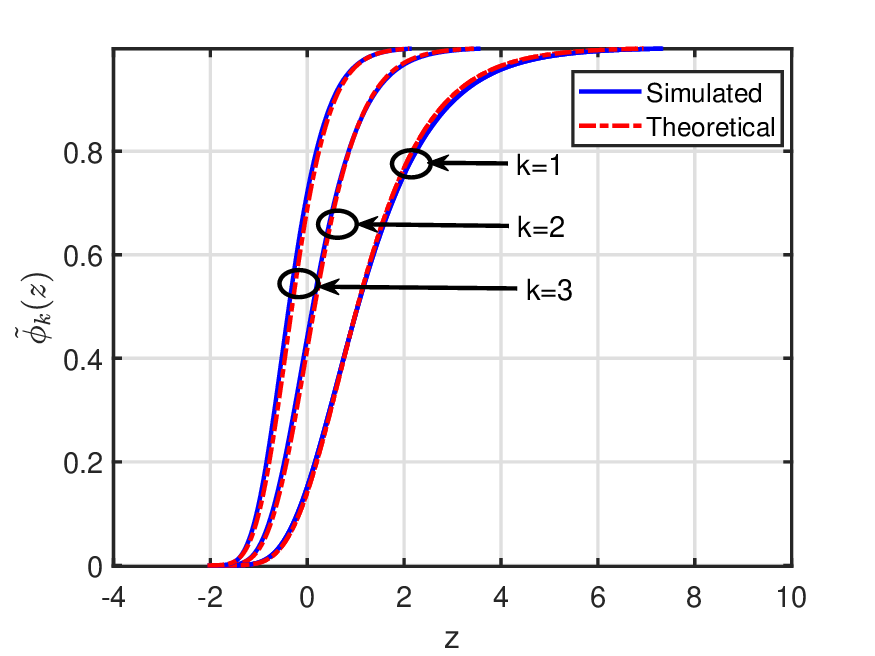}
			\caption{Asymptotic CDF of $\tilde{Z}_{(M-k+1:M)}$ for $M=20$.}
			\label{k_th_max_inid_norm}
		\end{minipage}
	\end{figure}
	
	In this section, we verify the results derived in Section \ref{sec_kth_max} using simulation experiments. Figure \ref{k_th_max_iid_norm} and figure \ref{k_th_max_inid_norm} compares the simulated and theoretical curves of the asymptotic CDF of $\tilde{Z}_{(M-k+1:M)}$ computed using the results derived in sections \ref{kth_max_iid_derivation} and \ref{kth_max_inid_derivation} respectively. Here, we have chosen $\sigma_m=2$ and $M=20$ for all the simulations, unless mentioned otherwise.  Figure \ref{k_th_max_iid_norm} demonstrates the case of i.i.d. RVs with $\nu_m=1$ $\forall \ m$ and figure \ref{k_th_max_inid_norm} demonstrates the case of i.n.i.d. RVs with $\nu_m=1$ for $1\leq m \leq \frac{M}{2}$ and $\nu_m=0.5$ for $\frac{M}{2} < m \leq M$.

	\begin{figure}[ht]
		\begin{minipage}[b]{0.45\linewidth}
			\centering
			\includegraphics[scale=0.45]{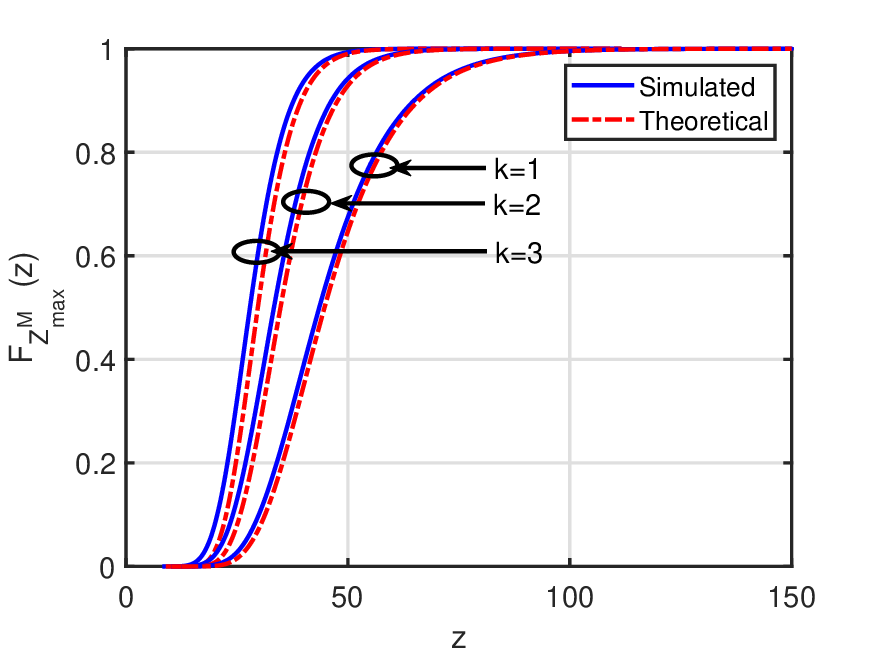}
			\caption{Asymptotic CDF of ${Z}_{(M-k+1:M)}$ for $M=30$.}
			\label{k_th_max_inid_unnorm}
		\end{minipage}
		\hspace{0.5cm}
		\begin{minipage}[b]{0.45\linewidth}
			\centering
			\includegraphics[scale=0.45]{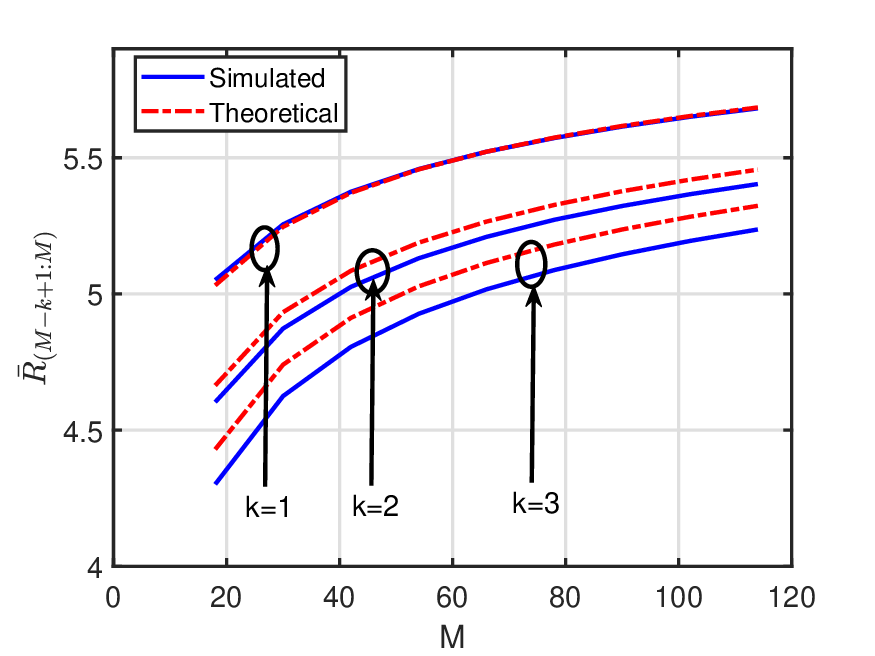}
			\caption{Average throughput for different values of $M$ and $k$.}
			\label{k_th_max_inid_unnorm_rate}
		\end{minipage}
	\end{figure}
	
	Next, in figure \ref{k_th_max_inid_unnorm}, we compare the simulated and theoretical CDF of $Z_{(M-k+1:M)}$ for $M=30$, $\nu_m=3; 1 \leq m \leq \frac{M}{3}$, $\nu_m=1; \frac{M}{3} < m \leq \frac{2M}{3}$ and $\nu_m=0.5;  \frac{2M}{3} < m \leq M$. In figure \ref{k_th_max_inid_unnorm_rate}, we compare the simulated and theoretical values of the average throughput $\bar{R}_{(M-k+1:M)}$ for different values of $M$ and $k$. Here, we use $\nu_m=2; 1 \leq m \leq \frac{M}{3}$, $\nu_m=1; \frac{M}{3} < m \leq \frac{2M}{3}$ and $\nu_m=0.5;  \frac{2M}{3} < m \leq M$. One can observe here that the convergence of the true values of the average throughput to the proposed values deteriorates as the value of $k$ increases. Furthermore, this convergence improves as the value of $M$ increases.  
	
	\begin{figure}[ht]
		\begin{minipage}[b]{0.45\linewidth}
			\centering
			\includegraphics[scale=0.45]{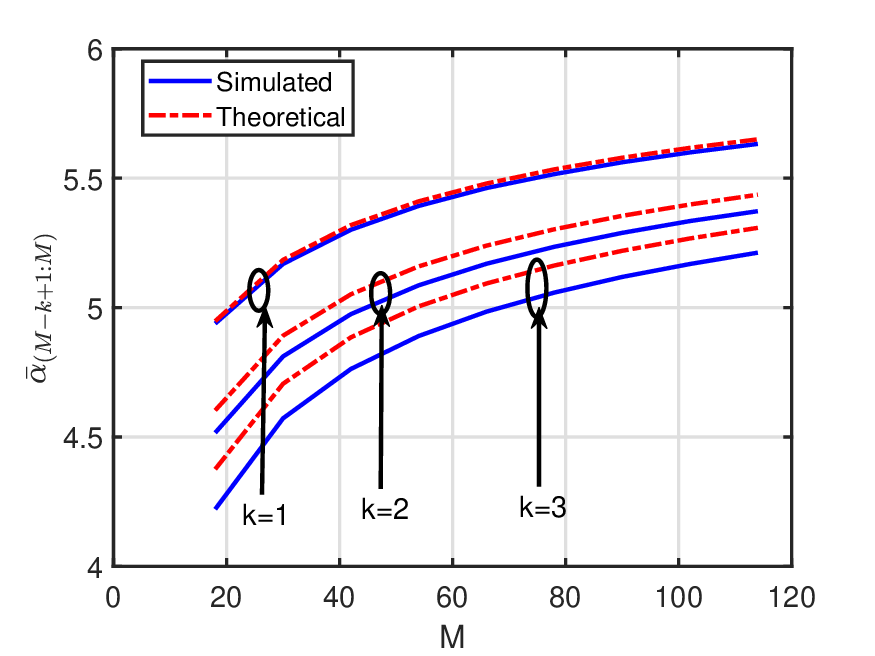}
			\caption{Achievable throughput for different values of $M$ and $k$.}
			\label{k_th_max_inid_unnorm_achvb_thrput}
		\end{minipage}
		\hspace{0.5cm}
		\begin{minipage}[b]{0.45\linewidth}
			\centering
			\includegraphics[scale=0.45]{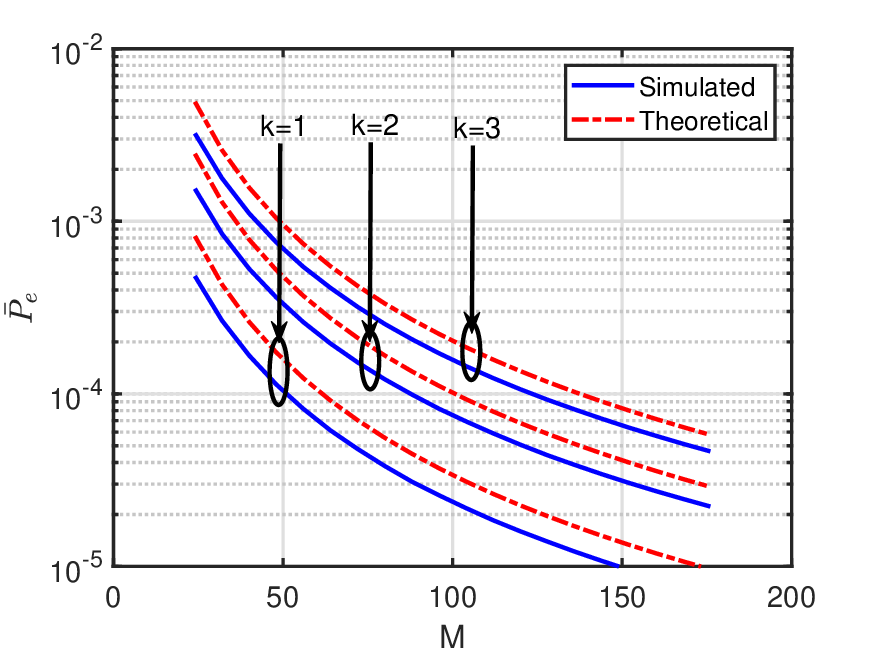}
			\caption{Asymptotic BER for different values of $M$ and $k$.}
			\label{k_th_max_inid_ber}
		\end{minipage}
	\end{figure}
	Figures \ref{k_th_max_inid_unnorm_achvb_thrput} and \ref{k_th_max_inid_ber} respectively compares the simulated and theoretical values of $\bar{\alpha}_{(M-k+1:M)}$ and $\bar{P}_{(M-k+1:M)}$ for different values of $M$ and $k$. Figure \ref{k_th_max_inid_unnorm_achvb_thrput} is generated using the same simulation set up for figure \ref{k_th_max_inid_unnorm_rate} and $\theta=1$. For figure \ref{k_th_max_inid_ber}, we generated SNR RVs with $\nu_m=1; 1 \leq m < \frac{3M}{4}$, $\nu_m=0.5; \frac{3M}{4} < m \leq M$, $C=0.25$ and $\rho=0.25$.
	
	\par {Next, we demonstrate one set of simulations to verify the utility of the asymptotic results for the applications proposed. We consider the system model in Application 1 for different number of eavesdroppers. Here we assume that $\sigma_s=-2$ dB, $\sigma=0.4$ and $\nu_m=0.07; 1 \leq m \leq \frac{M}{2}$, $\nu_m=0.05; \frac{M}{2} < m \leq \frac{3M}{4}$ and $\nu_m=0.01;  \frac{3M}{4} < m \leq M$. The CDF of the simulated values of $\gamma_{E,M:M}$ and the proposed asymptotic CDF are shown in Fig \ref{application_1_fig2}. The results show that the proposed asymptotic CDF characterises the exact CDF of the maximum very closely, especially for large values of $M$.}
	\begin{figure}[ht]
		\centering
		\includegraphics[scale=0.45]{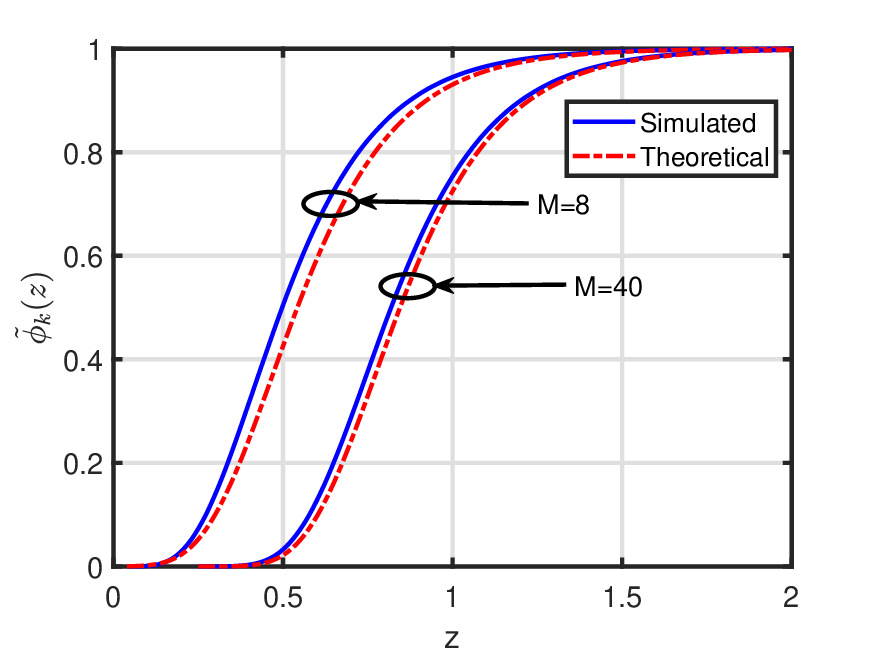}
		\caption{Asymptotic CDF of $\gamma_{E,M:M}$ for different values of $M$.}
		\label{application_1_fig2}
	\end{figure}
	
	\color{black}
	\section{Conclusions}\label{conclusion}
	In this work, we first characterized the distribution of the normalized $k$-th maximum of a sequence of i.n.i.d. RVs following the non-central chi-square distribution with two dof. Furthermore, we derived simple expressions for the corresponding values of asymptotic outage probability, average throughput, achievable throughput, and the average BEP. Our simulations showed that the proposed distribution accurately characterizes the exact distribution of maximum SNR, even for moderate values of $M$. We also discuss the utility of the proposed asymptotic results in simplifying the performance analysis of two IoT systems supported by UAV nodes. The proposed results is also demonstrated to be useful for studying an SC receiver. Moreover, the proposed results are general and can be used for any application involving performance analysis and resource allocation of systems where the order statistics of i.n.i.d. non-central chi-square RVs are involved. 
	\color{black}
	
	\begin{appendices}
		
		\section{Proof for Theorem \ref{thm_kth_max}} \label{proof_main}
		From Theorem \ref{thm_order_stat}, if we can identify normalising constants $a_M$ and $b_M$ such that the conditions in (\ref{ua1}), (\ref{ua2}) and (\ref{ua3}) are satisfied for $F_m(z)=F_{Z_m}(z)$, then we can characterise the distribution of the normalised $k$-th maximum $\tilde{Z}_{M-k+1:M}$. Mejzlers theorem \cite[Chapter 5]{de2007extreme} gives specific conditions on the normalising constants $a_M$ and $b_M$ such that the uniformity assumptions (\ref{ua1}) and (\ref{ua2}) are satisfied. Using these results, we assume that there exist sequences $a_M$ and $b_M$ such that
		\begin{align}
		\mid \log a_M \mid + \mid b_M \mid \ \to \infty \ \text{as} \ M \to \infty
		\label{norm_contn1}
		\end{align}
		and
		\begin{equation}
		\begin{array}{l}
		\frac{a_{M+1}}{a_{M}} \rightarrow 1 \\
		\frac{\left(b_{M+1}-b_{M}\right)}{a_{M}} \rightarrow 0,
		\end{array}
		\label{norm_contn2}
		\end{equation}
		are true. Next for $a_M$ and $b_M$ satisfying (\ref{norm_contn1}) and (\ref{norm_contn2}), we evaluate the limit in (\ref{ua2}) and hence try to identify closed form expressions for $a_M$ and $b_M$. In the following subsections, we present the steps to identify $a_M$ and $b_M$ for different special cases.
		
		\subsubsection{Identical Channel Parameters} \label{kth_max_iid_derivation}
		We begin with the most simple case where the fading channel parameters are identical across all the $m$ links where $m \in \{1,\cdots,M\}$. Hence, we have $\nu_m=\nu$ $\forall \ m$. Thus, (\ref{ua2}) for this case can be evaluated as
		\begin{equation}
		\tilde{u}(z) = \lim\limits_{M \to \infty} M \mathcal{Q}_1 \left(\frac{\nu}{\sigma}, \frac{\sqrt{a_M z+b_M}}{\sigma} \right) .
		\end{equation}
		We further make the assumption that $b_M \to \infty$ as $M \to \infty$. Hence the second argument of the marcum $Q$ function also grows to infinity as $M \to \infty$. Note that $\mathcal{Q}_{1}(\alpha, \beta)$ is related to the Gaussian $Q$ function as $\beta \rightarrow \infty, $  as given below \cite[(A 27)]{nuttall}:
		\begin{equation}
		\mathcal{Q}_1({\alpha,\beta}) \approx \sqrt{\frac{\beta}{\alpha}} \  Q(\beta-\alpha) \overset{(a)}{=}\sqrt{\frac{\beta}{\alpha}} \frac{1}{2} \operatorname{erfc}\left(\frac{\beta-\alpha}{\sqrt{2}}\right),
		\label{marcum_q1}
		\end{equation}
		where the equality in (a) is obtained by expressing the $Q(.)$ function in terms of the complementary error function $erfc(.)$ \cite[Eq.12.26]{molisch2012wireless}. The complementary error function has the following asymptotic expansion \cite{erfc}: 
		\begin{equation}
		\operatorname{erfc}(x) = \frac{1}{\sqrt{\pi}x} \exp({-x^{2}})\left(1+\mathcal{O}\left(\frac{1}{x^{2}}\right)\right)  \ |x| \rightarrow \infty,
		\end{equation}
		using which the asymptotic expansion for the $Q$ function can be written as 
		\begin{equation}
		Q(x) = \frac{x^{-1}}{\sqrt{2 \pi}} \exp\left(\frac{-x^2}{2} \right).
		\label{q}
		\end{equation}
		Using (\ref{q}), (\ref{marcum_q1}) can be written as 
		\begin{equation}
		\mathcal{Q}_{1}(\alpha,\beta) = \sqrt{\frac{\beta}{\alpha}}  \frac{(\beta-\alpha)^{-1}}{\sqrt{2 \pi}} \exp\left(\frac{-(\beta-\alpha)^2}{2} \right) \overset{(b)}{\approx}  \frac{1}{\sqrt{2 \pi \alpha \beta}}\exp\left(\frac{-(\beta-\alpha)^2}{2} \right).
		\label{marcum_q_approx}
		\end{equation} where the approximation in (b) is valid for large $\beta$. Now to evaluate $\tilde{u}(z)$, we have
		\begin{align}
		\alpha = \frac{\nu}{\sigma} \ \text{and} \ 
		\beta  = \frac{\sqrt{a_Mz +b_M}}{\sigma}.
		\end{align}
		Thus, we have
		\begin{align}
		(2 \pi \alpha \beta)^{\frac{-1}{2}}  = \left(2 \pi \frac{\nu}{\sigma} \frac{\sqrt{a_Mz+b_M}}{\sigma} \right)^{\frac{-1}{2}}=
		\left(\frac{2 \pi \nu}{\sigma^2} \right)^{\frac{-1}{2}} \exp \left(\frac{-1}{4} \log \left(a_M z+b_M \right) \right).
		\end{align}
		Similarly, we have
		\begin{align}
		\exp\left(\frac{-(\beta-\alpha)^2}{2}\right)  = \exp \left(\frac{-(a_Mz+b_M)}{2\sigma^2}-\frac{\nu^2}{2 \sigma^2} + \frac{\nu \sqrt{a_Mz+b_M
		}}{\sigma^2} \right).
		\end{align}
		Now for very large $b_M$, we can approximate the above expression as
		\begin{equation}
		\exp \left(\frac{-(\beta-\alpha)^2}{2}\right) \approx \exp\left(\frac{-a_Mz-b_M}{2 \sigma^2}\right) \exp\left(\frac{-\nu^2}{2 \sigma^2}\right) \exp \left( \frac{\nu \sqrt{b_M}}{\sigma^2}\right).
		\end{equation}
		Hence, we need to identify at least one pair of normalising constants $a_M$ and $b_M$ such that
		\begin{align}
		\tilde{u}(z) = \lim\limits_{M \to \infty} M & \left(\frac{2 \pi \nu}{\sigma^2} \right)^{\frac{-1}{2}} \exp \left(\frac{-1}{4} \log \left(a_M z+b_M \right) \right) \exp\left(\frac{-a_Mz}{2 \sigma^2}\right) \exp\left(\frac{-\nu^2}{2 \sigma^2}\right) \nonumber \\ & \times  \exp\left(\frac{-b_M}{2 \sigma^2} \right) \exp \left( \frac{\nu \sqrt{b_M}}{\sigma^2}\right) < \infty
		\end{align} and is a non-trivial function of $z$. Note that if we choose $a_M$ and $b_M$ to be of the form, 
		\begin{align}
		a_M &= 2 \sigma^2 \ \text{and} \ 
		b_M = a_M \left\lbrace \log(M)-c_0 \log(\log(M)) + c_1 \sqrt{\log(M)} - c_2\right \rbrace,
		\label{norm_const_iid}
		\end{align}where $c_0,c_1,c_2$ are constants with respect to $M$,
		we can make the following observations. 
		\begin{align}
		\tilde{u}(z) = \lim\limits_{M \to \infty} M & \left(\frac{2 \pi \nu}{\sigma^2} \right)^{\frac{-1}{2}} \exp \left(\frac{-1}{4} \log \left(a_M z+b_M \right) \right) \exp\left(-z\right) \exp\left(\frac{-\nu^2}{2 \sigma^2}\right) \nonumber \\ & \times  \exp\left(-\log(M)+c_0 \log(\log(M)) - c_1 \sqrt{\log(M)} + c_2 \right)  \nonumber  \\ & \times  \exp\left(\frac{\nu \sqrt{2}}{\sigma} \left(\log(M)-c_0 \log(\log(M)) + c_1 \sqrt{\log(M)} - c_2 \right)^{\frac{1}{2}} \right).
		\end{align}
		Now, we have
		\begin{align}
		\lim\limits_{M \to \infty} \frac{a_Mz+b_M}{\log(M)} =\lim\limits_{M \to \infty} (2 \sigma^2) \frac{z+\log(M)-c_0 \log(\log(M)) + c_1 \sqrt{\log(M)} -c_2}{\log(M)} = 2 \sigma^2.
		\end{align}
		Thus, if we choose $c_0=\frac{1}{4}$,
		\begin{equation}
		\lim\limits_{M \to \infty} \exp \left(\frac{-\log(a_Mz+b_M)}{4}  \right) \exp\left(c_0\log(\log(M)) \right) = \exp\left(\frac{-\log(2 \sigma^2)}{4} \right) = \left( 2 \sigma^2 \right)^{\frac{-1}{4}}.
		\end{equation}
		Similarly, note that
		\begin{align}
		& \lim\limits_{M \to \infty} \left(\log(M)-c_0 \log(\log(M)) + c_1 \sqrt{\log(M)} - c_2 \right)^{\frac{1}{2}} - \left(\log(M) \right)^{\frac{1}{2}} \\ = & \lim\limits_{M \to \infty} \left(\left(\sqrt{\log(M)} + \frac{c_1}{2} \right)^2 - \frac{c_1^2}{4} - c_0 \log(\log(M)) - c_2 \right)^{\frac{1}{2}} - \left(\log(M) \right)^{\frac{1}{2}}.
		\end{align}
		Let $p=\sqrt{\log(M)}$, $q=\frac{c_1}{2}$. Note that $p \to \infty$ as $M \to \infty$. After substituting $p$ and $q$, the limit we have to evaluate is given by
		\begin{equation}
		\lim\limits_{M \to \infty} \left((p+q)^2-q^2-c_0 \log(p^2) - c_2 \right)^{\frac{1}{2}} - p = \lim\limits_{M \to \infty} p \left\lbrace 1+\frac{2q}{p} - \frac{c_0 \log(p^2)}{p^2} - \frac{c_2}{p^2} \right \rbrace^{\frac{1}{2}}-p \label{step21}.
		\end{equation}
		Now, for large $p$, (\ref{step21}) can be approximated as $ \lim\limits_{M \to \infty} p \left\lbrace 1+\frac{2q}{p} \right \rbrace^{\frac{1}{2}}-p $. Using the binomial expansion of $\left\lbrace 1+\frac{2q}{p} \right \rbrace^{\frac{1}{2}}$, this limit can be evaluated as follows:
		\begin{equation}
		\lim\limits_{M \to \infty} p \left\lbrace 1+\frac{2q}{p} \right \rbrace^{\frac{1}{2}}-p = \lim\limits_{M \to \infty}   p \left \lbrace 1 + \frac{2q}{2p} - \frac{4q^2}{8 p^2} + \mathcal{O}(p^{-3})\right \rbrace -p = q = \frac{c_1}{2}.
		\end{equation}
		Thus, choosing $c_1 = \frac{\nu \sqrt{2}}{\sigma}$, we have $\lim\limits_{M \to \infty} \exp \left(-c_1 \sqrt{\log(M)} \right) \times \exp \left(\frac{\nu \sqrt{b_M}}{\sigma^2} \right) = \exp \left( \frac{\nu^2}{\sigma^2} \right)$. Thus, with $c_0=\frac{1}{4}$ and 
		$c_1 = \frac{\nu \sqrt{2}}{\sigma}$, we have
		\begin{align}
		\tilde{u}(z) = \lim\limits_{M \to \infty} \color{orange}\underbrace{\color{black}M} & \left(\frac{2 \pi \nu}{\sigma^2} \right)^{\frac{-1}{2}} \color{violet}\underbrace{\color{black}\exp \left(\frac{-1}{4} \log \left(a_M z+b_M \right) \right)} \color{black}{\color{black}\exp\left(-z\right)} \exp\left(\frac{-\nu^2}{2 \sigma^2}\right) \nonumber \\ & \times \color{orange}\underbrace{\color{black}\exp(-\log(M))} \color{violet}\underbrace{\color{black}\exp\left( \frac{\log(\log(M))}{4}\right)} \color{black}{\color{black}\exp(c_2)}\color{green}\underbrace{\color{black} \exp\left(- c_1 \sqrt{\log(M)} ) \right) } \nonumber \\ & \times  \color{green}\underbrace{\color{black}\exp\left(\frac{\nu \sqrt{2}}{\sigma} \left(\log(M)-c_0 \log(\log(M)) + c_1 \sqrt{\log(M)} - c_2 \right)^{\frac{1}{2}} \right)}\color{black}.
		\end{align}
		That is, $\tilde{u}(z)   = \left(\frac{2 \pi \nu}{\sigma^2} \right)^{\frac{-1}{2}} \color{violet}\underbrace{\color{black}\left(2 \sigma^2 \right)^{\frac{-1}{4}}}\color{black} \exp(-z) \exp \left(\frac{-\nu^2}{2 \sigma^2} \right) \exp\left(c_2 \right)\color{green}\underbrace{\color{black}\exp \left(\frac{\nu^2}{\sigma^2} \right)}\color{black} $. Thus, choosing $c_2 = \frac{1}{2} \log \left( \frac{2 \sqrt{2} \pi \nu}{\sigma} \exp\left(\frac{-\nu^2}{\sigma^2} \right)\right)$
		, we have $ \tilde{u}(z) = \exp(-z)$.
		
		\subsubsection{Non-identical LOS Components}\label{kth_max_inid_derivation}
		Next, we move on to the case where the LOS components of the received signal are non-identical and the NLOS components are identical. Thus we have, $\{\nu_m\}_{m=1}^M$ are non-identical. Furthermore, we assume that there are only finite possible values that $\nu_m$ can take, i.e $\nu_m \in \{\nu_1,\nu_2,\cdots,\nu_K\}$ for all $m \in \{1,\cdots,M\}$ and $K$ finite. Let, $M_i := \sum\limits_{m=1}^M \mathbb{I}_{\nu_m=\nu_i}, \ 1\leq i \leq K$, i.e $M_i$ represents the number of times $\nu_i$ occurs among the $M$ values of $\{\nu_m\}_{m=1}^M$. In this case, we need to identify $a_M$ and $b_M$ such that the following limit is finite and a non-trivial function of $z$.
		\begin{equation}
		\tilde{u}(z) = \lim\limits_{M \to \infty} \sum\limits_{m=1}^M \mathcal{Q}_1 \left(\frac{\nu_m}{\sigma}, \frac{\sqrt{a_Mz+b_M}}{\sigma} \right).
		\end{equation}
		Now, using the same approximations we used in the last derivation, for large values of $\beta$, we have
		\begin{equation}
		\mathcal{Q}_1(\alpha_m,\beta) \approx \left(2 \pi \alpha_m \beta \right)^{\frac{-1}{2}} \exp \left(\frac{-(\beta-\alpha_m)^2}{2} \right),
		\end{equation}
		where $\alpha_m  = \frac{\nu_m}{\sigma}$ and $\beta = \frac{\sqrt{a_Mz+b_M}}{\sigma}$. Thus, $\tilde{u}(z)$ can be computed as 
		\begin{align}
		\tilde{u}(z)   = \lim\limits_{M \to \infty} & {\exp \left(\frac{-1}{4} \log \left(a_M z+b_M \right) \right)} \exp\left(\frac{-a_Mz}{2 \sigma^2}\right) {\exp\left(\frac{-b_M}{2 \sigma^2}\right)} \nonumber \\ & \times  {\sum\limits_{m=1}^M \left(\frac{2\pi \nu_m}{\sigma^2} \right)^{\frac{-1}{2}} \exp\left(\frac{-\nu_m^2}{2 \sigma^2}\right) \exp \left( \frac{\nu_m \sqrt{b_M}}{\sigma^2}\right)}.
		\end{align}
		This can be further simplified as
		\begin{align}
		\tilde{u}(z)   = \lim\limits_{M \to \infty} & {\exp \left(\frac{-1}{4} \log \left(a_M z+b_M \right) \right)} \exp\left(\frac{-a_Mz}{2 \sigma^2}\right) {\exp\left(\frac{-b_M}{2 \sigma^2}\right)} \nonumber \\ & {\sum\limits_{i=1}^K M_i \left(\frac{2\pi \nu_i}{\sigma^2} \right)^{\frac{-1}{2}} \exp\left(\frac{-\nu_{i}^2}{2 \sigma^2}\right) \exp \left( \frac{\nu_i \sqrt{b_M}}{\sigma^2}\right)}.
		\label{uz_inid}
		\end{align}
		Next, we propose the following choice of normalising constants:
		\begin{align}
		a_M = 2 \sigma^2 \ \text{and} \ b_M  = a_M \times c_M \ \text{where}
		\label{norm_const_inid}
		\end{align}
		\begin{align}
		c_M = \left(\log(\tilde{M}) - \frac{\log(\log(\tilde{M}))}{4}  + \frac{\tilde{\nu} \sqrt{2}}{\sigma} (\log(\tilde{M}))^{\frac{1}{2}}  - \frac{1}{2} \log \left( \frac{2 \sqrt{2}\pi \tilde{\nu}}{\sigma}\exp\left(\frac{-\tilde{\nu}^2}{\sigma^2} \right) \right)\right).
		\end{align}
		Now, for the above choice of normalising constants we analyse (\ref{uz_inid}) for different values of $i$ and hence decide what values to choose for $\tilde{M}$ and $\tilde{\nu}$. Here we have, \begin{align}
		& \lim\limits_{M \to \infty} \exp(-z) M_i \exp \left(\frac{-1}{4} \log \left(a_M z+b_M \right) \right)  \left(\frac{2\pi \nu_i}{\sigma^2} \right)^{\frac{-1}{2}} \exp \left(\frac{-\nu_i^2}{2 \sigma^2}  \right)\exp \left( \frac{\nu_i \sqrt{b_M}}{\sigma^2}\right) \nonumber  \\&  \exp\left(-\log(\tilde{M}) + \frac{\log(\log(\tilde{M}))}{4}  - \frac{\tilde{\nu} \sqrt{2}}{\sigma} (\log(\tilde{M}))^{\frac{1}{2}}  + \frac{1}{2} \log \left( \frac{2 \sqrt{2}\pi \tilde{\nu}}{\sigma}\exp\left(\frac{-\tilde{\nu}^2}{\sigma^2} \right) \right)\right) .
		\end{align}
		We can rewrite the above expression as 
		\begin{align}
		& \lim\limits_{M \to \infty} \exp(-z) \frac{M_i}{\tilde{M}} \exp \left(\frac{-1}{4} \log \left(\frac{a_Mz+b_M}{\log(\tilde{M})} \right) \right) \left(\frac{2 \pi \nu_i}{\sigma^2}
		\right)^{\frac{-1}{2}} \left(\frac{2 \sqrt{2}\pi \tilde{\nu}}{\sigma} \right)^{\frac{1}{2}} \exp\left(\frac{-\nu_i^2}{2 \sigma^2} \right) \nonumber  \\ & \times
		\exp \left(\frac{-\tilde{\nu}^2}{2 \sigma^2} \right) \exp \left(\frac{\nu_i  \sqrt{b_M}}{\sigma^2} \right)  \exp \left(\frac{-\sqrt{2}\tilde{\nu}}{\sigma}(\log(\tilde{M}))^{\frac{1}{2}} \right).
		\end{align}
		Now, suppose we choose $\tilde{M}$ such that $\tilde{M} \to \infty$ as $M \to \infty$, then we have
		\begin{equation}
		\lim\limits_{M \to \infty} \exp \left(\frac{-1}{4} \log \left(\frac{a_Mz+b_M}{\log(\tilde{M})} \right) \right) = (2 \sigma^2)^{\frac{-1}{4}}.
		\end{equation}
		Similarly,
		\begin{equation}
		\exp \left(\frac{\nu_i \sqrt{2} \sqrt{c_M}}{\sigma} \right)  \exp \left(\frac{-\sqrt{2}\tilde{\nu}}{\sigma}(\log(\tilde{M}))^{\frac{1}{2}} \right) = \exp \left(\frac{\sqrt{2}}{\sigma} (\nu_i\sqrt{c_M}-\tilde{\nu}\sqrt{\log(\tilde{M})})\right).
		\end{equation}
		Here, we first analyse the following term
		\begin{align}
		\lim\limits_{M \to \infty}   \nu_i\left(\log(\tilde{M}) - \frac{\log(\log(\tilde{M}))}{4} + \frac{\tilde{\nu} \sqrt{2}}{\sigma} (\log(\tilde{M}))^{\frac{1}{2}} - \frac{1}{2} \log \left(c_2 \right) \right)^{\frac{1}{2}} - \tilde{\nu} \sqrt{\log(\tilde{M})}   ,
		\end{align}
		where $c_2=\frac{2 \sqrt{2}\pi \tilde{\nu}}{\sigma}\exp\left(\frac{-\tilde{\nu}^2}{\sigma^2} \right) $. Thus, we have
		\begin{equation}
		\lim\limits_{M \to \infty} \nu_i \left(\left(\sqrt{\log(\tilde{M})} + \frac{\tilde{\nu} }{\sqrt{2}\sigma} \right)^2 - \frac{\tilde{\nu}{^2}}{2 \sigma^2} -  \frac{\log(\log(\tilde{M}))}{4} - \frac{1}{2} \log \left(c_2 \right)  \right)^{\frac{1}{2}} - \tilde{\nu} \left(\log(\tilde{M}) \right)^{\frac{1}{2}}.
		\end{equation}
		
		Let $p=\sqrt{\log(\tilde{M})}$, $q=\frac{\tilde{\nu}}{\sqrt{2} \sigma}$ and $c_0=\frac{1}{4}$. Note that $p \to \infty$ as $\tilde{M} \to \infty$. After substituting $p$ and $q$, the limit we have to evaluate is given by
		\begin{equation}
		\lim\limits_{M \to \infty} \nu_i\left((p+q)^2-q^2-c_0 \log(p^2) - c_2 \right)^{\frac{1}{2}} - \tilde{\nu}p =    \lim\limits_{M \to \infty} \nu_i p \left\lbrace 1+\frac{2q}{p} - \frac{c_0 \log(p^2)}{p^2} - \frac{c_2}{p^2} \right \rbrace^{\frac{1}{2}}-\tilde{\nu} p  
		\label{step111}
		\end{equation}
		Now, for large $p$, (\ref{step111}) can be approximated as $\lim\limits_{M \to \infty} \nu_i p \left\lbrace 1+\frac{2q}{p} \right \rbrace^{\frac{1}{2}}-\tilde{\nu} p$. Using the binomial expansion of $\left\lbrace 1+\frac{2q}{p} \right \rbrace^{\frac{1}{2}}$, this limit can be evaluated as:
		\begin{equation}
		\lim\limits_{M \to \infty} \nu_i p \left\lbrace 1+\frac{2q}{p} \right \rbrace^{\frac{1}{2}}-\tilde{\nu} p = \lim\limits_{M \to \infty}   \nu_i p \left \lbrace 1 + \frac{2q}{2p} - \frac{4q^2}{8 p^2} + \mathcal{O}(p^{-3})\right \rbrace - \tilde{\nu} p = (\nu_i-\tilde{\nu})p+{\nu_i}q.
		\end{equation}
		Thus, we have
		\begin{equation}
		\lim\limits_{M \to \infty} \exp(-z) \sum\limits_{i=1}^2 \frac{M_i}{\tilde{M}} \sqrt{\frac{\tilde{\nu}}{\nu_i}} \exp \left(\frac{-\nu_i^2}{2 \sigma^2} \right) \exp \left(\frac{-\tilde{\nu}^2}{2 \sigma^2} \right) \exp\left(\frac{\sqrt{2}}{\sigma} \left((\nu_i-\tilde{\nu})\sqrt{\log(\tilde{M})}+\frac{\nu_i\tilde{\nu}}{\sqrt{2} \sigma} \right)\right).
		\end{equation}
		Now, let us choose $\tilde{\nu}$ to be the largest among $\{\nu_1,\cdots,\nu_K\}$ such that
		\begin{equation}
		\lim\limits_{M \to \infty} \frac{M_i}{\tilde{M}} \exp \left(\frac{\sqrt{2\log(\tilde{M}{})}(\nu_i-\tilde{\nu})}{\sigma} \right) < \infty \quad  \text{for} \quad \nu_i\neq \tilde{\nu},
		\label{m_condn1}
		\end{equation}
		where $\tilde{M} = \sum\limits_{m=1}^M \mathbb{I}_{\nu_m=\tilde{\nu}}$. Note that with finite choices for $\nu_i$ and $M \to \infty$, there exists at least one $\nu_i; \ i \in \{1,\cdots,K\}$ that satisfies the condition in (\ref{m_condn1}).
		In other words, suppose $K=2$, $M_1=\mathcal{O}(f_1(M))$ and $M_2=\mathcal{O}(f_2(M))$, then $\tilde{\nu}=\nu_i$ such that $f_i(M) \geq f_j(M)$ for all $j\neq i$. Now if $f_1(M)=f_2(M)$, we choose the largest $\nu_i$ among them as $\tilde{\nu}$. Thus we have,
		\begin{align}
		\lim\limits_{M \to \infty} & \exp(-z)  \frac{M_i}{\tilde{M}} \exp \left(\frac{-1}{4} \log \left(\frac{a_Mz+b_M}{\log(\tilde{M})} \right) \right) \left(\frac{2 \pi \nu_i}{\sigma^2}
		\right)^{\frac{-1}{2}} \left(\frac{2 \sqrt{2}\pi \tilde{\nu}}{\sigma} \right)^{\frac{1}{2}} \exp\left(\frac{-\nu_i^2}{2 \sigma^2} \right) \nonumber \\ & \times 
		\exp \left(\frac{-\tilde{\nu}^2}{2 \sigma^2} \right) \exp \left(\frac{\nu_i  \sqrt{b_M}}{\sigma^2} \right)  \exp \left(\frac{-\sqrt{2}\tilde{\nu}}{\sigma}(\log(\tilde{M}))^{\frac{1}{2}} \right) \\ 
		= \exp(-z) & \frac{M_i}{\tilde{M}} \sqrt{\frac{\tilde{\nu}}{\nu_i}} \exp \left(\frac{-\nu_i^2-\tilde{\nu}^2}{2 \sigma^2} \right)  \exp\left(\frac{\sqrt{2}}{\sigma} \left((\nu_i-\tilde{\nu})\sqrt{\log(\tilde{M})}+\frac{\nu_i\tilde{\nu}}{\sqrt{2} \sigma} \right)\right) < \infty,
		\end{align}for all values of $i$.
		Hence,
		\begin{equation}
		\tilde{u}(z) = \exp(-z) \sum\limits_{i=1}^K \frac{M_i}{\tilde{M}} \sqrt{\frac{\tilde{\nu}}{\nu_i}} \exp \left(\frac{-\nu_i^2-\tilde{\nu}^2}{2 \sigma^2} \right)  \exp\left(\frac{\sqrt{2}}{\sigma} \left((\nu_i-\tilde{\nu})\sqrt{\log(\tilde{M})}+\frac{\nu_i\tilde{\nu}}{\sqrt{2} \sigma} \right)\right),
		\label{uz_tilde_inid}
		\end{equation}
		where $\tilde{\nu}$ is the largest among $\{\nu_1,\cdots,\nu_K\}$ such that (\ref{m_condn1}) is satisfied. Let 
		
		\begin{equation}
		p(M_1,\cdots,M_K,\nu_1,\cdots,\nu_K,\sigma):=\sum\limits_{i=1}^K \frac{M_i}{\tilde{M}} \sqrt{\frac{\tilde{\nu}}{\nu_i}} \exp \left(\frac{-\nu_i^2-\tilde{\nu}^2}{2 \sigma^2} \right)  \exp\left(\frac{\sqrt{2}}{\sigma} \left((\nu_i-\tilde{\nu})\sqrt{\log(\tilde{M})}+\frac{\nu_i\tilde{\nu}}{\sqrt{2} \sigma} \right)\right),
		\end{equation} then we have $\tilde{u}(z) = \exp(-z) \times p(M_1,\cdots,M_K,\nu_1,\cdots,\nu_K,\sigma)$.
		Furthermore, note that for the case of very large $\tilde{M}$, $\sqrt{\log (\tilde{M})} \to \infty$ and hence for $\nu_i<\tilde{\nu}$, we have
		\begin{equation}
		\lim\limits_{M \to \infty}  \exp \left (\frac{\sqrt{2}}{\sigma}\left((\nu_i-\tilde{\nu})\sqrt{\log(\tilde{M})}+ \frac{\nu_2 \nu_1}{\sqrt{2}\sigma} \right) \right) = 0.
		\end{equation} Similarly, from our choice of $\tilde{\nu}$ and $\tilde{M}$, there cannot exist another $\nu_i> \tilde{\nu}$ such that $M_i >> \tilde{M}$. Hence, for very large $\tilde{M}$, we have $\lim\limits_{M \to \infty} \frac{M_i}{\tilde{M}} = 0$ and hence $\tilde{u}(z)$ again takes the form $\exp(-z)$. \\
		For both the cases discussed above when $k=1$, we have $\tilde{\phi}_k(z)=\exp(-\exp(-z))$ which is a non-degenerate distribution function. Note that according to \cite{barakat2002limit}, this is sufficient to ensure the convergence of $\tilde{\phi}_k(z)$ for all values of $k>1$. Thus, the distribution of the normalised $k$-th maximum for both the two cases discussed above is given by
		\begin{equation}
		\tilde{\phi}_k(z) = \sum\limits_{m=0}^{k-1} \frac{\exp(-mz)}{m!} \exp(-\exp(-z)) = \frac{\Gamma(k,\exp(-z))}{\Gamma(k)},
		\label{k_th_max_exprsn}
		\end{equation}where $\Gamma(.,.)$ is the upper incomplete Gamma function \cite{upper_gamma} and  $\Gamma(.)$ is the Gamma function \cite{gamma}.

		\section{Proof for Theorem \ref{rate_series_thm}} \label{proof_rate}
		Note that the expectation we need to evaluate (for all the special cases) is given by
		\begin{equation}
		R_{(M-k+1:M)} = \mathbb{E}\left[\log_2 \left(1+ \gamma_s Z_{(M-k+1:M)} \right) \right]
		\label{rate_expect}
		\end{equation}where the RV $Z_{(M-k+1:M)}$ has the following pdf:
		\begin{equation}
		f_{Z_{(M-k+1:M)}}(z) = \frac{p^k}{a_M \Gamma(k)} \exp \left(\frac{-k(z-b_M)}{a_M} \right) \exp \left( -p \exp\left(\frac{-(z-b_M)}{a_M} \right) \right).
		\end{equation}
		In the following lines we try to evaluate the expectation in (\ref{rate_expect}) using steps very similar to derivation in \cite[Appendix C]{kalyani2012analysis}. Note that the expectation in (\ref{rate_expect}) can be evaluated as 
		\begin{equation}
		R_{(M-k+1:M)} = \int\limits_0^{\infty} \log_2(1+z)  \  f_{Z_{(M-k+1:M)}}(z) \ dz.
		\label{intg_expr}
		\end{equation}
		Now, substituting $y=\frac{z-b_M}{a_M}$, we have
		\begin{equation}
		R_{(M-k+1:M)} = \frac{p^k}{\Gamma(k)} \int\limits_{\frac{-b_M}{a_M}}^{\infty} \log_2\left(1+\gamma_s(a_My+b_M) \right) \exp(-ky) \exp(-p \exp(-y)) \ dy. 
		\end{equation}
		Using the series expansion for $\exp(-p \exp(-y))$, we can rewrite the previous expression as 
		\begin{equation}
		R_{(M-k+1:M)} = \frac{p^k}{\Gamma(k)} \int\limits_{\frac{-b_M}{a_M}}^{\infty} \log_2\left(1+\gamma_s(a_My+b_M) \right) \exp(-ky) \sum\limits_{n=0}^\infty \frac{(-1)^n p^n e^{-ny}}{n!} \ dy.
		\label{rate_int_first}
		\end{equation}
		Now, the integral and the summation in (\ref{rate_int_first}) can be interchanged. This argument can be proved by  applying the Lebesgue dominated convergence theorem which is stated below:
		\begin{theorem}
			Suppose $f_{m}: M \rightarrow \mathbb{R}$ is a sequence of measurable functions, and if we can find another sequence of non-negative measurable functions $d_{m}: M \rightarrow \mathbb{R}$ such that $\left|f_{m}(x)\right| \leq d_{m}(x), \forall m$ and almost all $x$ and $\sum_{m=0}^{\infty} d_{m}(x)$
			converges and $\sum_{m=0}^{\infty} \int d_{m}(x)<\infty,$ then
			$$
			\int_{M} \sum_{m=0}^{\infty} f_{m}(x) d x=\sum_{m=0}^{\infty} \int_{M} f_{m}(x) d x
			$$
		\end{theorem}
		We can apply Lebesgue theorem by taking $f_{n}(y)=$ $\frac{(-1)^{n}p^n}{n!} \log_2 \left(1+\gamma_s (b_{M}+ a_{M} y)\right) e^{-(k+n) y}$ and $d_{n}(y)=\frac{p^n}{n!} \log_2 \left(1+\gamma_s (b_{M}+ a_{M} y)\right) e^{-(k+n) y}$. Let us first verify if $\sum\limits_{n=0}^\infty d_n(y)$ converges. Applying the ratio test to the series $\sum\limits_{n=0}^\infty d_n(y)$, we have
		\begin{equation}
		\lim\limits_{n \to \infty} \left| \frac{d_{n+1}}{d_n} \right| = \frac{p}{(n+1)} e^{-y} = 0.
		\end{equation}
		Thus, we conclude that the series converges. Next, we verify the finiteness of $\sum\limits_{n=0}^\infty \int d_n(y)$, where $h_n:=\int d_n(y)$. Here we have,
		\begin{equation}
		\sum\limits_{n=0}^\infty h_n=   \sum\limits_{n=0}^\infty \frac{p^n}{n!} \int\limits_{\frac{-b_M}{a_M}}^\infty   \log_2 \left(1+\gamma_s (b_{M}+ a_{M} y)\right) e^{-(k+n) y} \ dy .
		\end{equation}
		Substituting, $1+\gamma_s (b_{M}+ a_{M} y)=x$, we have
		\begin{align}
		\sum\limits_{n=0}^\infty h_n  & = \sum\limits_{n=0}^\infty \frac{p^n}{n!} \frac{\exp \left(\frac{(k+n)(1+\gamma_s b_M)}{\gamma_s a_M} \right)}{\gamma_s a_M} \int\limits_1^\infty \log_2(x) \exp \left(\frac{-(k+n)x}{\gamma_s a_M} \right)  \ dx, \\
		& = \sum\limits_{n=0}^\infty \frac{p^n}{n! (k+n)} \exp \left(\frac{(k+n)(1+\gamma_s b_M)}{\gamma_s a_M} \right)  \Gamma \left(0,\frac{k+n}{a_M \gamma_s} \right).
		\end{align}
		Again applying ratio test on the above series, we have
		\begin{equation}
		\lim\limits_{n \to \infty} \frac{h_{n+1}}{h_n} =  \lim\limits_{n \to \infty }  \frac{p (k+n)}{n+1 (k+n+1)} \exp \left(\frac{1+\gamma_s b_M}{\gamma_s a_M} \right) \ \frac{\Gamma\left(0,\frac{k+n+1}{a_M \gamma_s} \right)}{\Gamma\left(0,\frac{k+n}{a_M \gamma_s} \right)}. 
		\end{equation}
		Using properties of the incomplete gamma function, one can note that $\frac{\Gamma\left(0,\frac{k+n+1}{a_M \gamma_s} \right)}{\Gamma\left(0,\frac{k+n}{a_M \gamma_s} \right)} < 1$. Hence,
		we have
		\begin{equation}
		\lim\limits_{n \to \infty} \frac{h_{n+1}}{h_n} = 0. 
		\end{equation}Thus we conclude that the series $\sum\limits_{n=0}^\infty h_n$ converges. This, in turn allows us to interchange the summation and integration in (\ref{rate_int_first}). Thus, we have
		\begin{equation}
		R_{(M-k+1:M)} = \frac{p^k}{\Gamma(k)} \sum\limits_{n=0}^\infty \frac{(-1)^n p^n}{n!} \int\limits_{\frac{-b_M}{a_M}}^\infty \log_2 \left(1+\gamma_s (b_{M}+ a_{M} y)\right) \ e^{-y(k+n)} \ dy.
		\end{equation}
		Again substituting $1+\gamma_s(a_My+b_M) = x$, we have
		\begin{equation}
		R_{(M-k+1:M)} = \frac{p^k}{\Gamma(k)} \sum\limits_{n=0}^\infty \frac{(-1)^n p^n}{n!} \frac{\exp \left(\frac{(k+n)(1+\gamma_sb_M)}{\gamma_s a_M} \right)}{\gamma_s a_M} \ \int\limits_1^\infty \log_2(x) \ \exp\left( \frac{-(k+n)x}{\gamma_s a_M} \right) \ dx.
		\end{equation}
		Evaluating the above integral we have
		\begin{equation}
		R_{(M-k+1:M)} = \frac{p^k}{\Gamma(k)} \sum\limits_{n=0}^\infty \frac{(-1)^n p^n}{n!} \frac{\exp \left(\frac{(k+n)(1+\gamma_sb_M)}{\gamma_s a_M} \right)}{\gamma_s a_M} \ \frac{a_M \gamma_s \Gamma\left(0,\frac{k+n}{a_M \gamma_s} \right)}{(k+n)}.
		\end{equation}
		Note that the above expression can be rewritten in terms of the exponential integral function using the relation, $\Gamma(0,x)=E_1(x)$ \cite[(5)]{exp_integral_1}. Thus, we have the expression in (\ref{sum_expr}).
	\end{appendices}
	
	\bibliographystyle{IEEEtran}
	\bibliography{reference}
\end{document}